\documentclass[envcountreset]{llncs}
\usepackage{amsfonts}
\usepackage{amsmath}

\usepackage{amsthm}
\usepackage{amssymb}
\usepackage{enumerate}
\usepackage[final,letterpaper]{hyperref}
\usepackage{tikz}

\newtheorem{thm}{Theorem}

\newtheorem{obs}{Observation}
\newtheorem{lem}[thm]{Lemma}
\newtheorem{conj}{Conjecture}

\newenvironment{defn}[1][Definition]{\begin{trivlist}
\item[\hskip \labelsep {\bfseries #1}]}{\end{trivlist}}

\newenvironment{proofof}[1]{\begin{trivlist}
\item[\hskip \labelsep {\itshape Proof of #1.}]}{\end{trivlist}}

\newcommand{\U}{\,\mathcal{U}}
\newcommand{\G}{\,\mathcal{G}}
\newcommand{\C}{\,\mathcal{C}}
\newcommand{\N}{\mathbb{N}}
\newcommand{\R}{\mathbb{R}}
\newcommand{\Z}{\mathbb{Z}}

\newcommand{\wh}{\widehat}

\begin{document}

\title{Self-stabilizing uncoupled dynamics\thanks{Research supported in part by National Science Foundation Grant No. CCF-1101690.}}
\author{
Aaron D. Jaggard\inst{1}
\and Neil Lutz\inst{2}
\and Michael Schapira\inst{3}
\and Rebecca N. Wright\inst{4}}

\institute{
U.S. Naval Research Laboratory, Washington, DC 20375, USA. \email{aaron.jaggard@nrl.navy.mil}
\and
Rutgers University, Piscataway, NJ 08854, USA. \email{njlutz@cs.rutgers.edu}
\and
Hebrew University of Jerusalem, Jerusalem 91904, Israel. \email{schapiram@huji.ac.il}
\and
Rutgers University, Piscataway, NJ 08854, USA. \email{rebecca.wright@rutgers.edu}}

\tocauthor{Aaron D.\ Jaggard, Neil Lutz, Michael Schapira, Rebecca N.\ Wright}

\date{}
\maketitle

\begin{abstract}
Dynamics in a distributed system are self-stabilizing if they are guaranteed to reach a stable state regardless of how the system is initialized. Game dynamics are uncoupled if each player's behavior is independent of the other players' preferences. Recognizing an equilibrium in this setting is a distributed computational task. Self-stabilizing uncoupled dynamics, then, have both resilience to arbitrary initial states and distribution of knowledge. We study these dynamics by analyzing their behavior in a bounded-recall synchronous environment. We determine, for every ``size'' of game, the minimum number of periods of play that stochastic (randomized) players must recall in order for uncoupled dynamics to be self-stabilizing.  We also do this for the special case when the game is guaranteed to have unique best replies. For deterministic players, we demonstrate two self-stabilizing uncoupled protocols. One applies to all games and uses three steps of recall. The other uses two steps of recall and applies to games where each player has at least four available actions. For uncoupled deterministic players, we prove that a single step of recall is insufficient to achieve self-stabilization, regardless of the number of available actions.
\end{abstract}

\section{Introduction}
Self-stabilization is a failure-resilience property that is central to distributed computing theory and is the subject of extensive research (see, e.g.,~\cite{Dole00} for a survey). It is characterized by the ability of a distributed system to reach a stable state from every initial state. Dynamic interaction between strategic agents is a central research topic in game theory (see, e.g.,~\cite{FudTir91,NRTV07}).  One area of interest is \emph{uncoupled dynamics}, in which each player's strategy is independent of the other players' payoffs~\cite{HarMas12}.  Here, we bring together these two research areas and initiate the study of \emph{self-stabilizing uncoupled dynamics} within the broader research agenda of  \emph{distributed computing with adaptive heuristics}~\cite{JaScWr11}. This work is a first step, and the same questions we answer here can be asked for a broad variety of dynamics and notions of convergence and equilibria. These directions, as well as a conjecture, are discussed in Section~\ref{sec:5}.

We focus our investigation on a bounded-recall, synchronous setting. We consider self-stabilization in a multi-agent distributed system in which, at each timestep, the agents act as strategic players in a game, simultaneously selecting \emph{actions} from their respective finite action sets to form an \emph{action profile}. The space of action profiles is relevant throughout this work, and we refer to its size as the \emph{size} of the game. We study the effects of \emph{bounded recall}, in which the state of this system at any time consists of the $r$ most recent action profiles, for some finite $r$. The stable states in $r$-recall systems necessarily have the same action profile in $r$ consecutive time steps. In our context, we want stable states that are robust to players acting selfishly---i.e., those where the repeated action profile is an equilibrium of the stage game. In this paper, we consider pure Nash equilibria (PNE). Thus, in our setting, dynamics self-stabilize for a given game if, from every starting state, players are guaranteed to converge to a PNE. For games without PNE, dynamics cannot self-stabilize in this sense.  Throughout this paper, we say that particular dynamics \emph{succeed} on a class of games if they self-stabilize for games in that class whenever a PNE exists.

Traditional study of convergence to equilibria in game dynamics makes various assumptions about the ``reasonableness'' of players' behavior, restricting them to always play the game in ways that are somehow consistent with their self-interest given their current knowledge. In contrast to these behavioral restrictions on the players, uncoupledness is an informational restriction, in that the players have no knowledge of each other's payoffs. In this situation, no individual player can recognize a PNE, so finding an equilibrium is a truly distributed task. 

If uncoupledness is the only restriction on the dynamics, then the players can find a PNE through a straightforward exhaustive search.  However, this changes when players' abilities to remember past actions is restricted.  In a continuous-time setting, Hart and Mas-Colell~\cite{HarMas03} showed that deterministic uncoupled dynamics fail to reach a stable state for some games that have PNE if the dynamics must be \emph{historyless}, i.e., if the state space of the system is identical to the action profile space of the game. This suggests the central question that we address:
\begin{quotation}\emph{On a given class of games, how much recall do uncoupled players need in order to self-stabilize whenever a PNE exists?  That is, when are there successful $k$-recall dynamics?}
\end{quotation}

This question was answered in part by Hart and Mas-Colell~\cite{HarMas06}, who showed that in a discrete-time setting, even when players are allowed randomness, no historyless uncoupled dynamics succeed on all two-player games where each player has three actions. Moreover, they showed that even for \emph{generic} games (where at every action profile each player has a unique ``best'' action), no historyless uncoupled dynamics succeed on games with three three-action players. They also gave positive results, proving that there are historyless uncoupled dynamics that succeed on all two-player generic games, and that if the players have 2-recall (i.e., they are allowed to see the two most recent action profiles), then over every action profile space there are stochastic uncoupled dynamics that succeed on all games.

\paragraph{Our results.}
We show in Section \ref{sec:3} that there exist historyless uncoupled dynamics that succeed on all two-player games with a two-action player and on all three-player generic games with a two-action player (Theorems \ref{thm:6} and \ref{thm:11}). In both cases, we prove that these results are tight, in that they do not hold for any larger size of game (Theorems \ref{thm:7} and \ref{thm:13}). Combined with the results of Hart and Mas-Colell~\cite{HarMas06}, this provides a complete characterization of the exact minimum recall needed, for any action profile space, for uncoupled dynamics to succeed on all games over that space and on generic games over that space. In Section \ref{sec:4}, turning to deterministic dynamics, we demonstrate 3-recall deterministic uncoupled dynamics that succeed on all games (Theorem \ref{thm:16}) and 2-recall deterministic uncoupled dynamics that succeed on all games in which every player has at least four actions (Theorem \ref{thm:17}). We also prove for every action profile space that no historyless deterministic uncoupled dynamics succeed on all games over that space (Theorem \ref{thm:18}).

\paragraph{Related work.}
There are rich connections between distributed computing and game theory, some of which are surveyed by Halpern~\cite{Hal07}.  Jaggard, Schapira, and Wright~\cite{JaScWr11} investigated convergence to pure Nash equilibria by game dynamics in asynchronous distributed systems.  Most closely related to our specific setting, Hart and Mas-Colell introduced the concept of uncoupled game dynamics~\cite{HarMas03}.  In addition to the results mentioned above, they also addressed convergence to mixed Nash equilibria by bounded-recall uncoupled dynamics~\cite{HarMas06}. Babichenko investigated the situation when the uncoupled players are finite-state automata, as well as \emph{completely uncoupled dynamics}, in which each player can see only the history of its own actions and payoffs \cite{Babi06,Babi12}. Hart and Mansour~\cite{HarMan10} analyzed the time to convergence for uncoupled dynamics. 

\section{Preliminaries}

We begin with definitions of the concepts used in the paper.

\paragraph{Games.} Let $n\in\N$ and $(k_1,...,k_n)\in\N^n$, with $n\geq2$ and each $k_i\geq 2$. A \emph{game} of \emph{size} $(k_1,...,k_n)$ is a pair $(A,U)$, where $A=A_1\times...\times A_n$ such that each $|A_i|=k_i$, and $U=(u_1,...,u_n)$ is an $n$-tuple of functions $u_i:A\to \R$. $A_i$ and $u_i$ are the \emph{action set} and \emph{utility function} of \emph{player} $i$. $\Delta(A_i)$, the probability simplex over $A_i$, is player $i$'s set of \emph{mixed actions}. When $n$ is small, we may describe a game $(A,U)$ as a $k_1$-\emph{by}-...-\emph{by}-$k_n$ game. Elements of $A$ are the \emph{(action) profiles} of the game, and $A$ is called the \emph{(action) profile space}. $\U(A)$ is the the class of all $U$ such that each $u_i$ takes $A_i$ as input, so $A\times\U(A)$ is the class of all games with profile space $A$. When $A$ is clear from context, we often identify the game with the utility function vector $U$.

Let $U\in\U(A)$. For $i\in\{1,...,n\}$ and $a=(a_1,...,a_n)\in A$, we say that player $i$ is $U$-\emph{best-replying} at $a$ if $u_i(a)\geq u_i((a_1,...,a^\prime_i,...,a_n))$ for every $a^\prime_i\in A_i$. We define the set of $U$-\emph{best-replies} for player $i$ at $a$,
\[BR^U_i(a)=\{a^\prime_i\in A_i:i\mbox{ is }U\mbox{-best-replying at }(a_1,...,a^\prime_i,...,a_n)\}.\]
We omit $U$ from this notation when the game being played is clear from context. A profile $p\in A$ is a \emph{pure Nash equilibrium}, abbreviated \emph{PNE}, for $U$ if every player $i\in\{1,...,n\}$ is best-replying at $p$.  An action $a_i\in A_i$ is \emph{weakly dominant} for player $i$ if $a_i\in BR_i(x)$ for every $x\in A$; it is \emph{strictly dominant} for player $i$ if $BR_i(x)=\{a_i\}$ for every $x\in A$.

A game $(A,U)\in A\times\U(A)$ is \emph{generic} if every player's best-replies are unique, i.e., if for every $a\in A$ and $i\in\{1,...,n\}$, $|BR^U_i(a)|$=1. For generic games $(A,U)$ we may abuse notation slightly by using $BR^U_i(a)$ to refer to this set's unique element.  $A\times\G(A)$ is the class of all generic games on $A$.
\paragraph{Dynamics.} We now consider the repeated play of a game. Let the profile at timestep $t\in\Z$ be $a^{(t)}=\big(a^{(t)}_1,...,a^{(t)}_n\big)\in A$. The \emph{stage game} $(A,U)\in A\times\U(A)$ is then played: each player $i$ simultaneously selects a new action $a^{(t+1)}_i$ by applying an $r$-\emph{recall} \emph{stationary} \emph{strategy} $f^U_i:A^r\to \Delta(A_i)$, where $r\in\N$ and $A^r$ is the Cartesian product of $A$ with itself $r$ times. A \emph{deterministic} r-recall stationary strategy mapping ranges over $A_i$ instead of $\Delta(A_i)$. The strategy $f^U_i$, which is \emph{stationary} in the sense that it does not depend on $t$, will take as input $(a^{(t-r+1)},...,a^{(t)})$, the $r$ most recent profiles. We call this $r$-tuple the \emph{state} at time $t$. The terms $1$-\emph{recall} and \emph{historyless} are interchangeable. A \emph{strategy vector} is an $n$-tuple $f^U=(f^U_1,...,f^U_n)$, where each $f^U_i$ is a strategy for player $i$. $\mathcal{F}(A)$ will denote the set of all strategy vectors for $A$.

A \emph{strategy mapping} for $A$ is a mapping $f:\U(A)\to\mathcal{F}(A)$ that assigns to each $U$ a strategy vector $f^U$. A strategy mapping $f$ is \emph{uncoupled} if the strategy it assigns each player depends only on that player's utility function and not, e.g., on the other players' payoffs. That is, there are mappings $f_1,...,f_n$ where each $f_i$ maps utility functions on $A$ to strategies for $A$, such that $f_i(u_i)\equiv f^U_i$ for $i=1,...,n$. If $f^U_i$ is stationary, deterministic, or $r$-recall for $i=1,...,n$, then $f^U$ is also. If every $f^U$ has any of those properties, then $f$ does also.

Now let $x=\big(x^{(1)},...,x^{(r)}\big)\in A^r$, and let $f^U$ be an $r$-recall strategy vector. For $T\geq r$, a \emph{partial} $f^U$-\emph{run} for $T$ steps starting from $x$ is a tuple of profiles $\big(a^{(1)},...,a^{(T+r)}\big)\in A^{T+r}$ such that $x=(a^{(1)},...,a^{(r)})$ and for every $r<t\leq T+r$,
\[\Pr\left(f^U\big(a^{(t-r)},...,a^{(t-1)}\big)=a^{(t)}\right)>0.\]
An $f^U$-\emph{run} is an infinite sequence of profiles $a^{(1)},a^{(2)},...$ such that every finite prefix is a partial $f^U$-run. We say that $y\in A^r$ is $f^U$-\emph{reachable} from $x\in A^r$ if there exist a $T\in\N$ and a partial $f^U$-run $\big(a^{(1)},...,a^{(T+r)}\big)$ such that  $x=(a^{(1)},...,a^{(r)})$ and $y=\big(a^{(T)},...,a^{(T+r)}\big)$. The state $x$ is an $f^U$-\emph{absorbing state} if for every $f^U$-run $a^{(1)},a^{(2)},...$ beginning from $x$, $\big(a^{(t+1)},...,a^{(t+r)}\big)=x$ for every $t\in\N$. Notice that any $f^U$-absorbing state $x=\big(a^{(1)},...,a^{(r)}\big)$ must have $a^{(1)}=...=a^{(r)}$. We omit the strategy vector from this notation when it is clear from context. The \emph{game dynamics} of $f$ consist of all pairs $(U,R)$ such that $R$ is an $f^U$-run.

\paragraph{Convergence.} A sequence of profiles $a^{(1)},a^{(2)},...$ \emph{converges} to a profile $a$ if there some $T\in \N$ such that $a^{(t)}=a$ for every $t\geq T$. If from every $x\in A^r$, some $f^U$-absorbing PNE is $f^U$-reachable, then $f$ \emph{self-stabilizes} on game $(A,U)$. We say that $f$ \emph{succeeds} on a game $U$ if $f$ self-stabilizes on $(A,U)$ or if $(A,U)$ has no PNE. Let $\C(A)$ be a class of games on $A$. If $f$ succeeds on every game $(A,U)\in A\times \mathcal{C}(A)$, then $f$ \emph{succeeds} on $\C(A)$.

Let  $A=A_1\times...\times A_n$ and $B=B_1\times...\times B_n$ be profile spaces of the same size, in the sense that there is some permutation $\pi$ on $\{1,...,n\}$ such that  $(|A_1|,...,|A_n|)=(|B_{\pi(1)}|,...,|B_{\pi(n)}|)$. Then we write $A\simeq B$. If $f$ succeeds on $\C(A)$, then there is a strategy mapping derived from $f$ that succeeds on $\C(B)$, simply by rearranging the players and bijectively mapping actions in each $A_i$ to actions in $B_{\pi(i)}$. This new strategy mapping retains any properties of $f$ that are of interest here (uncoupledness, $r$-recall, stationarity, and determinism). For this reason we define
\[\C(|A_1|,...,|A_n|)=\bigcup_{B\simeq A} \C(B),\]
and we say that $f$ succeeds on $\C(|A_1|,...,|A_n|)$ if $f$ succeeds on $\C(B)$ for some $B\simeq A$. For example, ``$f$ succeeds on $\G(2,3)$'' means ``$f$ self-stabilizes on every generic 2-by-3 game with a PNE (up to renaming of actions).''

\section{Stochastic uncoupled dynamics}\label{sec:3}

In this section, we determine, for every profile space $A$, the minimum $r\in\N$ such that an uncoupled $r$-recall stationary strategy mapping exists that succeeds on all games $(A,U)\in A\times\U(A)$ or all generic games $(A,U)\in A\times\G(A)$. Hart and Mas-Colell~\cite{HarMas06} proved that $2$-recall is sufficient to succeed on all games, 1-recall is sufficient to succeed on generic two-player games, and that 1-recall is not sufficient to succeed on all games, or even all generic games. We state these results in the present setting.
\begin{thm}[Hart and Mas-Colell~\cite{HarMas06}]\label{thm:1}
For any profile space $A$, there exists an uncoupled $2$-recall stationary strategy mapping that succeeds on all games $(A,U)$.
\end{thm}
\begin{thm}[Hart and Mas-Colell~\cite{HarMas06}]\label{thm:2}
There is no uncoupled historyless stationary strategy mapping that succeeds on all 3-by-3 games, or on all 3-by-3-by-3 generic games.
\end{thm}
\begin{thm}[Hart and Mas-Colell~\cite{HarMas06}]\label{thm:3}
For any two-player profile space $A$, there is an uncoupled historyless stationary strategy mapping that succeeds on all generic games $(A,U)$.
\end{thm}
We now describe the strategy mapping given in the proof of Theorem \ref{thm:3}. Notice that for a historyless stationary strategy mapping, the state space is exactly the profile space, so the terms \emph{state} and \emph{profile} are interchangeable in this context.
\begin{defn}
For any $n$-player profile space $A$, the \emph{canonical} historyless uncoupled stationary strategy mapping for $A$ is $h:\U(A)\to\mathcal{F}(A)$, defined as follows. Let $U=(u_1,...,u_n)\in\U(A)$. Then $h(U)=(h^U_1,...,h^U_n)$, where for $i\in\{1,...,n\}$, $h^U_i:A\to A_i$ is given by
\begin{align*}
\Pr\left(h^U(a_i)=a_i\,|\,a_i\in BR_i(a)\right)&=1\\
\Pr\left(h^U(a_i)=b_i\,|\,a_i\not\in BR_i(a)\right)&=1/k_i,
\end{align*}
\end{defn}
for all $a_i,b_i\in A_i$. That is, if player $i$ is already best replying, then it will continue to play the same action. Otherwise, $i$ will play an action chosen uniformly at random from its action set.

In their proof of Theorem \ref{thm:2}, Hart and Mas-Colell make the following observation.
\begin{obs}[Hart and Mas-Colell~\cite{HarMas06}]\label{obs:4}
Suppose $f$ is an uncoupled historyless stationary strategy mapping for profile space $A$ and $f$ succeeds on all generic games $(A,U)$.  Then two conditions hold for every game $(A,U)$ and $a=(a_1,...,a_n)\in A$. First, if player $i$ is best-replying at $a$, then $Pr(f^U_i(a)=a_i)=1$. Second, if player $i$ is not best replying at $a$, then $\Pr(f^U_i(a)=a^\prime_i)>0$ for some $a^\prime_i\in A_i\smallsetminus\{a_i\}$.
\end{obs}
Informally, no player can move when it is best-replying, and each player must move w.p.p. whenever it is not best-replying. The first condition guarantees that every PNE is an absorbing state; the second guarantees that no non-PNE is an absorbing state. Implicit in the same proof is the fact that $h$ is at least as ``powerful'' as any other historyless uncoupled strategy mapping.

\begin{obs}[Hart and Mas-Colell~\cite{HarMas06}]\label{obs:5}
If any historyless uncoupled stationary strategy mapping succeeds on $\U(A)$ or on $\G(A)$, then $h$ succeeds on that class.
\end{obs}

\subsection{Stochastic dynamics for $\U(A)$}
We now describe the profile spaces in which there are uncoupled historyless strategy mappings that succeed on every game, or equivalently (by Observation \ref{obs:5}), the $A$ for which $h$ succeeds on $\U(A)$. A proof that $h$ succeeds on 2-by-$k$ games is given in the appendix. It proceeds by simple case checking but may be a useful warmup for working with these dynamics.

\begin{thm}\label{thm:6}
For every two-player profile space $A$ in which one player has only two actions, $h$ succeeds on all games $(A,U)$.
\end{thm}

It turns out that $2$-by-$k$ profile spaces are the only ones where $h$ succeeds on all games.
\begin{thm}\label{thm:7}
Let $A$ be a profile space. Unless $A$ has only two players and one of those players has only two actions, no historyless uncoupled stationary strategy mapping succeeds on all games $(A,U)$.
\end{thm}
We give three lemmas that will be used in the proof of Theorem \ref{thm:7}. Their full proofs are in the appendix. Informally, Lemma \ref{lem:8} says that additional actions do not make a profile space any ``easier'' in this context; the players will need at least as much recall to succeed on all games in the larger space. The proof relies on a type of reduction in which the players take advantage of a strategy mapping for a larger game by ``pretending'' to play the larger game. Whenever player $i$ plays $k_i$, all players guess randomly whether $i$ would have played $k_i$ or $k_i+1$ in the larger game.

\begin{lem}\label{lem:8}
Let $n\geq 2$, $k_1,...,k_n\geq2$, and $i\in\{1,...,n\}$. If $h$ succeeds on $\U(k_1,...,k_i+1,...,k_n)$, then $h$ succeeds on $\U(k_1,...,k_i,...,k_n)$.
\end{lem}
Lemma \ref{lem:9} tells us that the same is true of adding players to the game. Its proof also uses a simple reduction. The players utilize the strategy mapping for the $(n+1)$-player game by behaving as if there is an additional player who never wishes to move. This preserves genericity, so the lemma also applies to the class of generic games.

\begin{lem}\label{lem:9}
Let $n\geq 2$ and $k_1,...,k_n,k_{n+1}\geq2$. If $h$ succeeds on $\U(k_1,...,k_n,k_{n+1})$, then $h$ succeeds on $\U(k_1,...,k_i,...,k_n)$. The same is true if we replace $\U$ with $\G$.
\end{lem}
Finally, Lemma \ref{lem:10} says that $h$ does not succeed on all $2$-by-$2$-by-$2$ games. An example is given in its proof of a game with a PNE where $h$ fails to converge.

\begin{lem}\label{lem:10}
No historyless uncoupled stationary strategy mapping succeeds on $\U(2,2,2)$.
\end{lem}

\begin{proofof}{Theorem \ref{thm:7}}
Let $A=A_1\times...\times A_n$. By Observation \ref{obs:5}, it suffices to show that $h$ does not succeed on $\U(|A_1|,...,|A_n|)$. Assume that $h$ does succeed on $\U(|A_1|,...,|A_n|)$. If $n=2$, $|A_1|,|A_2|>2$, and $h$ succeeds on $\U(k_1,k_2)$, then by repeatedly applying Lemma \ref{lem:8}, $h$ succeeds on $\U(3,3)$. This contradicts Theorem \ref{thm:2}. Now suppose that $n\geq 3$. If $h$ succeeds on $\U(|A_1|,...,|A_n|)$, then by repeatedly applying Lemma \ref{lem:9}, $h$ succeeds on $\U(|A_1|,|A_2|,|A_3|)$. So by repeatedly applying Lemma \ref{lem:8}, $h$ succeeds on $\U(2,2,2)$. This contradicts Lemma \ref{lem:10}.
\end{proofof}

\subsection{Stochastic dynamics for $\G(A)$}
We now turn to generic games and to describing the class of profile spaces $A$ for which there exist historyless uncoupled strategy mappings that succeed on $\G(A)$. Theorem \ref{thm:3} tells us that $h$ succeeds on two-player generic games. In fact, $h$ also succeeds on three-player generic games where one player has only two options.

\begin{thm}\label{thm:11}
Let $A$ be a three-player profile space such that one player has only two actions. Then $h$ succeeds on all generic games $(A,U)$.
\end{thm}
The proof of this theorem relies partially on an analogy between a $k$-by-$\ell$-by-2 generic game and a $k\ell$-by-2 game that might not be generic. This requires the following technical lemma showing that under $h$, two players in a generic game sometimes behave similarly to a single player.

\begin{lem}\label{lem:12}
Let $k,l\in\N$, and let $U\in\G(k,\ell)$ be a game in which neither player has a strictly dominant action. For every $a,b\in A$ such that $a$ is not a PNE for $U$, $b$ is $h^U$-reachable from $a$.
\end{lem}

\begin{proofof}{Theorem \ref{thm:11}}
Let $A=\{1,...,k\}\times\{1,...,\ell\}\times\{0,1\}$ for some $\ell,k\in\N$.
Let $U\in\G(A)$ and $a=(a_1,a_2,a_3)\in A$. All PNE are absorbing states under $h$, so it will suffice to show there is some PNE that is $h^U$-reachable from $a$.

Let $A^\prime=\{1,...,k\}\times\{1,...,\ell\}$, and consider the games $U^0=(u^0_1,u^0_2)$ and $U^1=(u^1_1,u^1_2)\in\G(A^\prime)$ defined by
\begin{align*}
u^0_i(x_1,x_2)&=u_i(x_1,x_2,0)\\
u^1_i(x_1,x_2)&=u_i(x_1,x_2,1)
\end{align*}
for every $x_1\in\{1,...,k\}$, $x_2\in\{1,...,\ell\}$, and $i\in\{0,1\}$. In this proof we will repeatedly use the fact that over any finite number of steps, w.p.p. player 3 doesn't move, so if $(y_1,y_2)\in A^\prime$ is $h^{U^0}$-reachable from $(x_1,x_2)\in A^\prime$, then $(y_1,y_2,0)\in A$ is $h^U$-reachable from $(x_1,x_2,0)\in A$, and similarly for $h^{U^1}$.

\begin{claim}
If either player has a strictly dominant action in $U^0$ or $U^1$, then some PNE is $h^U$-reachable from $a$.
\end{claim}
This claim is proved in the appendix. Thus we may assume that neither player has a strictly dominant action in $U^0$ or in $U^1$. Consider a two-player game $\wh{U}=(\wh{u}_1,\wh{u}_2)$ on $\wh{A}=(\{1,...,k\}\times\{1,...,\ell\})\times\{0,1\}$ given by
\begin{align*}
\wh{u}_1(x)&=\left\{\begin{array}{ll}
1\;&\;\mbox{if }(x_1,x_2)\mbox{ is a PNE for }U^{x_3}\\
0\;&\;\mbox{otherwise}
\end{array}\right.\\
\wh{u}_2(x)&=u_3((x_1,x_2,x_3)),
\end{align*}
for every $x=((x_1,x_2),x_3)\in\wh{A}$. Note that unlike $U$, this game is not necessarily generic. By Theorem \ref{thm:6}, some PNE $\wh{p}=((p_1,p_2),p_3)$ for $\wh{U}$ is $h^{\wh{U}}$-reachable from $\wh{a}=((a_1,a_2),a_3)$.

Now let $\wh{x}=((x_1,x_2),x_3)$ and $\wh{y}=((y_1,y_2),y_3)\in\wh{A}$ such that w.p.p. $\wh{y}=h^{\wh{U}}(\wh{x})$. If $x_3\neq y_3$, then $x_3\not\in BR^{\wh{U}}_2(\wh{x})$, so $x_3\neq BR^U_3(x)$. Thus w.p.p. $h^U(x)=(x_1,x_2,y_3)$. Since $BR^U_3(x)\neq x_3 \neq y_3$ and $|A_3|=2$, we must have $BR^U_3(x)=y_3$, so if $(x_1,x_2)$ is a PNE for $U^{y_3}$, then $(x_1,x_2,y_3)$ is a PNE for $U$. Otherwise, by Lemma \ref{lem:12} $(y_1,y_2)$ is $h^{U^{x_3}}$-reachable from $(x_1,x_2)$, so $y=(y_1,y_2,y_3)$ is $h^U$-reachable from $(x_1,x_2,y_3)$ and therefore from $x$.

Applying this to the each step on the path by which $\wh{p}$ is $h^{\wh{U}}$-reachable from $\wh{a}$, we see that either $p=(p_1,p_2,p_3)$  (which is a PNE for $U$) is $h^U$-reachable from $a$, or some other PNE for $U$ is encountered in this process and thus $h^U$-reachable from $a$.
\end{proofof}

In fact, two-player and $2$-by-$k$-by-$\ell$ are the only sizes of generic games on which $h$ always succeeds.

\begin{thm}\label{thm:13}
Let $A$ be a profile space. If $A$ has more than three players, or if every player has more than two actions, then no historyless uncoupled stationary strategy mapping succeeds on all generic games $(A,U)$.
\end{thm}
Before proving this theorem, we present two lemmas whose full proofs are in the appendix. Lemma \ref{lem:14} says that $h$ does not succeed on all 2-by-2-by-$k$-by-$\ell$ generic games. It is proved by giving an example of such a game.

\begin{lem}\label{lem:14}
For every $k,\ell\geq 2$, $h$ does not succeed on $\G(2,2,k,\ell)$.
\end{lem}
Lemma \ref{lem:15} says that $h$ doesn't succeed on all three-player generic games in which all players have at least three actions. This is demonstrated by simple modifications of the 3-by-3-by-3 game used by Hart and Mas-Colell in their proof of Theorem \ref{thm:2}.

\begin{lem}\label{lem:15}
For every $k_1,k_2,k_3\geq 3$, $h$ does not succeed on $\G(k_1,k_2,k_3)$
\end{lem}

\begin{proofof}{Theorem \ref{thm:13}}
By Observation \ref{obs:5}, if suffices to show that $h$ does not succeed on $\G(|A_1|,...,|A_n|)$. Assume for contradiction that $h$ does succeed on $\G(|A_1|,...,|A_n|)$. If $n=3$ and $h$ succeeds on $\G(|A_1|,|A_2|,|A_3|)$, then by Lemma \ref{lem:15} we cannot have $|A_1|,|A_2|,|A_3|>2$. If $n=4$ and $h$ succeeds on $\G(|A_1|,...,|A_4|)$, then by Lemma \ref{lem:14} there are distinct $i,j,k\in\{1,2,3,4\}$ such that $|A_i|,|A_j|,|A_k|>2$. But by Lemma \ref{lem:9}, $h$ succeeds on $\G(|A_i|,|A_j|,|A_k|)$, contradicting lemma \ref{lem:15}. If $n>4$ and $h$ succeeds on $\G(|A_1|,...,|A_n|)$, then by repeatedly applying Lemma \ref{lem:14}, $h$ succeeds on $\G(|A_1|,...,|A_4|)$, which we have already shown to be impossible.
\end{proofof}

\section{Deterministic uncoupled dynamics}\label{sec:4}

Both $h$ and the strategy mapping used by Hart and Mas-Colell~\cite{HarMas06} to prove Theorem \ref{thm:1} are variations on random search. For deterministic dynamics, an exhaustive search requires more structure, and the challenge for deterministic players in short-recall uncoupled dynamics is in keeping track of their progress in the search.

\subsection{Positive results}
We show that there are successful 3-recall deterministic dynamics by using repeated profiles to coordinate.
\begin{thm}\label{thm:16}
For every profile space A, there exists a deterministic uncoupled 3-recall stationary strategy mapping that succeeds on all games $(A,U)$.
\end{thm}
\begin{proof}
Let $n\geq 2$, $k_1,...,k_n\geq 2$, and $A=\{1,...,k_1\}\times...\times\{1,...,k_n\}$. It suffices to show that such a strategy mapping exists for $\U(A)$.
Let $\sigma:A\to A$ be a cyclic permutation on the profiles. We write $\sigma_i(a)$ for the action of player $i$ in $\sigma(a)$. Let $f:\U(A)\to\mathcal{F}(A)$ be the strategy mapping such that, for every game $U\in\U(A)$, player $i\in\{1,...,n\}$, and state $x=(a,b,c)\in A^3$,
\[f^U_i(x)=\left\{\begin{array}{ll}
c_i\;&\;\mbox{if }b=c\mbox{ and }c_i\in BR_i(c)\\
\min BR_i(c)\;&\;\mbox{if }b=c\mbox{ and }c_i\not\in BR_i(c)\\
\sigma_i(a)\;&\;\mbox{if }a=b\neq c\\
c_i\;&\;\mbox{otherwise}.
\end{array}\right.\]
Informally, the players use repetition to keep track of which profile is the current ``PNE candidate'' in each step. If a profile has just been repeated, then it is the current candidate, and each player plays a best reply to it, with a preference against moving. If the players look back and see that some profile $a$ was repeated in the past but then followed by a different profile, they infer that $a$ was rejected as a candidate and move on by playing $a$'s successor, $\sigma(a)$. Otherwise the players repeat the most recent profile, establishing it as the new candidate. We call these three types of states \emph{query}, \emph{move-on}, and \emph{repeat} states, respectively. Here ``query'' refers to asking each player for one of its best replies to $b$.

Let $U\in\U(A)$ be a game with at least one PNE. We wish to show that $f^U$ guarantees convergence to a PNE. Let $x=(a,b,c)\in A^3$, and let $y$ be the next state $(b,c,f^U(x))$. If $x$ is a repeat state, then $y=(b,c,c)$, which is a query state. If $x$ is a move-on state, then $b\neq c$, and $y=(b,c,\sigma(a))$. If $c=\sigma(a)$, then this is a query state; otherwise, it's a repeat state, which will be followed by the query state $(c,\sigma(a),\sigma(a))$. Thus every non-query state will be followed within two steps by a query state.

Now let $x=(a,b,b)\in A^3$ be a query state, and let $y$ and $z$ be the next two states. If $b$ is a PNE, then $y=(b,b,b)$, which is an absorbing state. Otherwise, $y=(b,b,c)$ for some $c\neq b$, so $y$ is a move-on state, which will be followed by a query state $(b,\sigma(b),\sigma(b))$ or $(c,\sigma(b),\sigma(b))$ within two steps. Let $p$ be a PNE for $U$. Since $\sigma$ is cyclic, $p=\sigma^r(b)$ for some $r\in\N$. So $(p,p,p)$ is reachable from $x$ unless $\sigma^s(b)$ is a PNE for some $s<r$. It follows that $f^U$ guarantees convergence to a PNE, so $f$ succeeds on $\U(A)$.
\end{proof}

Recall that Lemma \ref{lem:8} says that in the stochastic setting, adding actions to a profile space $A$ does not make success on $\U(A)$ any easier. In light of that result, it is perhaps surprising that we can improve on the above bound when every player has sufficiently many actions.

\begin{thm}\label{thm:17}
If $A$ is a profile space in which every player has at least four actions, then there exists a $2$-recall deterministic uncoupled stationary strategy mapping that succeeds on all games $(A,U)$.
\end{thm}
\begin{proof}
Let $n\geq 2$, $k_1,...,k_n\geq 4$, and $A=\{1,...,k_1\}\times...\times\{1,...,k_n\}$. It suffices to show that such a strategy mapping exists for $\U(A)$.

Define a permutation $\sigma:A\to A$ such that for every $a\in A$, $\sigma(a)$ is $a$'s lexicographic successor. Formally, $\sigma(a)=(\sigma_1(a),...,\sigma_n(a))$ where for $i=1,...,n-1$,
\[\sigma_i(a)=\left\{\begin{array}{ll}
a_i+1\bmod{k_i}\;&\;\mbox{if }a_j=k_j\mbox{ for every }j\in\{i+1,...,n\}\\
a_i\;&\;\mbox{otherwise},
\end{array}\right.
\]
and $\sigma_n(a)=a_n+1\bmod k_n$. Observe then that $\sigma$ is cyclic, and for each player $i$ and $a\in A$, we have
\[\sigma_i(a)-a_i \bmod{k_i}\in\{0,1\}.\]

We now describe a strategy mapping $f:\U(A)\to\mathcal{F}(A)$. To each $U\in\U$, $f$ assigns the strategy vector $f^U$ defined as follows.
At state $x=(a,b)\in A^2$, $f^U$ differentiates between three types of states, each named according to the event it prompts:
\begin{itemize}
\item \emph{move-on}: If $a\neq b$ and $a_j-b_j\bmod{k_j}\in \{0,1\}$ for every $j\in\{1,...,n\}$, then the players ``move on" from $a$, in the sense that each player $i$ plays $\sigma_i(a)$, giving $f^U(x)=\sigma(a)$.
\item \emph{query}: If $b_j-a_j\bmod{k_j}\in \{0,1,2\}$, then we ``query" each player's utility function to check whether it is $U$-best-replying at $b$. Each player $i$ answers by playing $b_i$ if it is best-replying and $b_i - 1\bmod{k_i}$ if it is not. So at query states,
\[f^U_i(x)=\left\{\begin{array}{ll}
b_i\;&\;\mbox{if }b_i\in BR_i(b)\\
b_i-1\bmod{k_i}\;&\;\mbox{otherwise},
\end{array}\right.
\]
for $i=1,...,n$.
\item \emph{repeat}: Otherwise, each player $i$ ``repeats" by playing $b_i$, giving $f^U(x)=b.$
\end{itemize}
Notice that because $k_1,...,k_n\geq4$, it is never the case that both $a_j-b_j\bmod{k_j}\in \{0,1\}$ and $b_j-a_j\bmod{k_j}\in \{0,1,2\}$. Thus the conditions for the \emph{move-on} and \emph{query} types are mutually exclusive, and the three state types are all disjoint.

The state following $x=(a,b)$ is $y=(b,f^U(x))$. If $x$ is a move-on state, then $y=(b,\sigma(a))$. Since for every player $i$, $a_i-b_i\bmod{k_i}\in\{0,1\}$ and $\sigma(a)_i-a_i\bmod{k_i}\in\{0,1\}$, we have $\sigma_i(a)-b_i\bmod{k_i}\in\{0,1,2\}$, so $y$ is a query state. If $x$ is instead a query state, then $b_i-f^U_i(x)\bmod{k_i}\in\{0,1\}$ for every player $i$, so $y$ is a move-on state unless $b=f^U(x)$, in which case $y=(b,b)$ is a query state. But if $b=f^U(x)$ and $x$ was a query state, then $b_i\in BR_i(b)$ for every player $i$, i.e., $b$ is a PNE. Finally, if $x$ is a repeat state, then $y=(b,b)$ is a query state.

Thus move-on states and repeat states are always followed by query states, and ask-all states are never followed by repeat states. We conclude that with the possible exception of the initial state, every state will be a move-on or query state, and no two consecutive states will be move-on states. In particular, some query state is reachable from every initial state.

For any query state $x=(a,b)$, $x$ will be followed by $(b,b)$ if and only if $b$ is a PNE, and $(b,b)$ is an absorbing state for every PNE $b$. If $b$ is not a PNE, then $x$ will be followed will be a move-on state $(b,c)$, for some $c\in A$. This will be followed by the \emph{query} state $(c,\sigma(b))$. Continuing inductively, since $\sigma$ is cyclic, unless the players converge to a PNE, they will examine every profile $v\in A$ with a query state of the form $(u,v)$. Thus for every game $U$ with at least one PNE, $f^U$ guarantees convergence to a PNE, i.e., $f$ succeeds on $\U(A)$.
\end{proof}

While there are deterministic uncoupled $2$-recall dynamics that succeed on at least some classes that require $2$-recall in the stochastic setting, historyless dynamics of this type fail on $\U(A)$ for every profile space $A$. A proof of the following theorem is given in the appendix.
\begin{thm}\label{thm:18}
For every profile space $A$, no deterministic uncoupled historyless stationary strategy mapping succeeds on all games $(A,U)$.
\end{thm}

\section{Future Directions}\label{sec:5}
It remains open to determine tight bounds on the minimum recall of successful deterministic uncoupled dynamics for every profile space, analogous to those given in Section~\ref{sec:3} for stochastic dynamics. In particular, we make the following conjecture.
\begin{conj}
There exists a profile space $A$ such that no deterministic uncoupled 2-recall strategy mapping succeeds on all games (A,U).
\end{conj}
The same questions answered in this work may naturally be asked for other important classes of games (e.g., symmetric games) and other equilibrium concepts, especially mixed Nash equilibrium. More generally, the resources (e.g., recall, memory) required by uncoupled self-stabilizing dynamics in asynchronous environments should be investigated.

\bibliography{ssud}
\bibliographystyle{abbrv}

\appendix
\setcounter{secnumdepth}{0}
\section{Appendix}
\newtheorem*{thm-6}{Theorem~\ref{thm:6}}
\begin{thm-6}
For every two-player profile space $A$ in which one player has only two actions, $h$ succeeds on all games $(A,U)$.
\end{thm-6}
\begin{proof}
Let $k\geq 2$. It suffices to show that $h$ succeeds on $\U(2,k)$. Let $A=\{1,2\}\times\{1,...,k\}$ and $U=(u_1,u_2)\in\U(A)$. Suppose that $U$ has at least one PNE, and recall that every PNE for $U$ is an $h^U$-absorbing state. Let $a=(a_1,a_2)\in A$, and consider four cases.
\begin{enumerate}
\item Player 1 is best-replying at $a$ and $a_1=p_1$ for some PNE $p=(p_1,p_2)$. Then either player $2$ is also best-replying and $a$ is a PNE, or $h^U_2(a)=p_2$ w.p.p., so $h^U(a)$ is a PNE w.p.p.
\item Player 1 is not best-replying at $a$ and there is no PNE $p$ such that $a_1=p_1$. Then w.p.p. $h^U_1(a)\neq a_1$ and $h^U_2(a)=a_2$. Since we assumed that $U$ has a PNE, $h^U(a)$ is then an instance of case 1.
\item Player 1 is best-replying at $a$ and there is no PNE $p$ such that $a_1=p_1$. Then player 2 is not best-replying at $a$, so w.p.p. $h^U_2(a)\in BR_2(a)$, but $h^U(a)$ cannot be a PNE since $h^U_1(a)=a_1$. Then player 1 is not best-replying at $h^U(a)$, i.e., $h^U(a)$ is an instance of case 2.
\item Player 1 is not best-replying at $a$ and $a_1=p_1$ for some PNE $p=(p_1,p_2)$. Then w.p.p. $h^U_1(a)\neq a_1$ and $h^U_2(a)=a_2$, in which case player 1 is best-replying at $h^U(a)=(h^U_1(a),a_2)$, since player 1 has only two actions. Then $h^U(a)$ is an instance of case 1 or 3.
\end{enumerate}
We conclude that from every state $a\in A$, some PNE for $U$ is $h^U$-reachable from $a$. Thus $h$ succeeds on $\U(A)$.
\end{proof}

\newtheorem*{lem-8}{Theorem~\ref{lem:8}}
\begin{lem-8}
Let $n\geq 2$, $k_1,...,k_n\geq2$, and $i\in\{1,...,n\}$. If $h$ succeeds on $\U(k_1,...,k_i+1,...,k_n)$, then $h$ succeeds on $\U(k_1,...,k_i,...,k_n)$.
\end{lem-8}
\begin{proof}
Let
\begin{align*}
A&=\{1,...,k_1\}\times...\times\{1,...,k_i\}\times...\times\{1,...,k_n\},\\
A^\prime&=\{1,...,k_1\}\times...\times\{1,...,k_i+1\}\times...\times\{1,...,k_n\}.
\end{align*}
Suppose that $h$ succeeds on $\U(A^\prime)$. For each $U=(u_1,...,u_n)\in\U(A)$, define another game $U^\prime=(u_1^\prime,...,u_n^\prime)\in\U(A^\prime)$ such that for every $j\in\{1,...,n\}$ and $a\in A$,
\begin{align*}
&u_j^\prime(a)=u_j(a),\\
&u_j^\prime(a_1,...,k_i+1,...,a_n)=u_j(a_1,...,k_i,...,a_n).
\end{align*}
Thus in $U^\prime$ every player is always indifferent to whether player $i$ plays $k_i$ or $k_i+1$.

We now define a strategy mapping $f$ for games on $A$. For every $U\in\U(A)$, $f^U$ is given by
\begin{align*}
&\Pr\left(f^U_j(a)=h^{U^\prime}_j(a_1,...,k_i+1,...,a_n)\,\big|\,a_i=k_i\right)=1/2,\\
&\Pr\left(f^U_j(a)=h^{U^\prime}_j(a)\,\big|\,a_i=k_i\right)=1/2,\mbox{ and}\\
&\Pr\left(f^U_j(a)=h^{U^\prime}_j(a)\,\big|\,a_i\neq k_i\right)=1,
\end{align*}
for every $a\in A$ and $j\neq i$. That is, whenever the players see that player $i$ has played $k_i$, each chooses independently at random to interpret that action either as $k_i$ or $k_i+1$, then plays the action prescribed by $h^{U^\prime}$. Player $i$ behaves similarly under $f$, but we have to ensure that it's never instructed to play $k_i+1$:
\begin{align*}
&\Pr\left(f^U_j(a)=\min\{k_i,h^{U^\prime}_i(a_1,...,k_i+1,...,a_n)\}\,\big|\,a_i=k_i\right)=1/2\\
&\Pr\left(f^U_j(a)=\min\{k_i,h^{U^\prime}_i(a)\,\big|\,a_i=k_i\right)=1/2\\
&\Pr\left(f^U_j(a)=h^{U^\prime}_i(a)\,\big|\,a_i\neq k_i\right)=1.
\end{align*}

Now fix $U=(u_1,...,u_n)\in\U(A)$, and assume that $U$ has at least one PNE $p\in A$. To see that $p$ is an absorbing state for $f^U$, we consider two cases. First, suppose that $p_i\neq k_i$. Then $p$ is also a PNE for $U^\prime$, hence $p$ is an absorbing state for $h^{U^\prime}$. So for $j\neq i$, $f^U_j(p)=h^{U^\prime}_j(p)=p_j$, and $f^U_i(p)=\min\{k_i,h^{U^\prime}_i(p)\}=\min\{k_i,p_i\}=p_i$. Now suppose instead that $p_i=k_i$. Then both $p$ and $p^\prime=(p_1,...,k_i+1,...,p_n)$ are stable states for $h(U^\prime)$. So for $j\neq i$,
\begin{align*}f^U_j(p)&\in \{h^{U^\prime}_j(p_1,...,k_i+1,...,p_n),f^{U^\prime}_j(p)\}\\
&= \{h^{U^\prime}_j(p^\prime),h^{U^\prime}_j(p)\}\\
&= \{p_j\},
\end{align*}
and
\begin{align*}
f^U_i(p)&\in\big\{\min\{k_i,h^{U^\prime}_i(p_1,...,k_i+1,...,p_n)\},\min\{k_i,h^{U^\prime}_i(p)\}\big\}\\
&\subseteq \big\{\min\{k_i,h^{U^\prime}_i(p^\prime)\},k_i,h^{U^\prime}_i(p)\big\}\\
&= \big\{\min\{k_i,p^\prime_i\},k_i,p_i\big\}\\
&=\{p_i\}.
\end{align*}
Thus $p$ is an absorbing state for $f^U$.

It remains to show that $f^U$ always reaches a PNE. Let $a\in A\subseteq A^\prime$. Since $U^\prime$ has a PNE and $h$ succeeds on $\U(A^\prime)$, $U^\prime$ has some PNE $q=(q_1,...,q_n)\in A^\prime$ such that $q$ is $h^{U^\prime}$-reachable from $a$. So for some $T\in\N$, theres is a partial $h^{U^\prime}$-run $a^{(0)},...,a^{(T)}$ such that $a^{(0)}=a$ and $a^{(T)}=q$. Since $q$ is a PNE for $U^\prime$, $q^\prime=(q_1,...,\min\{q_i,k_i\},...,q_n)$ is a PNE for both $U$ and $U^\prime$.

Now let $b^{(0)},...,b^{(T)}$ be a partial $f^U$-run such that $b^{(0)}=a$. Suppose, for some $0\leq t<T$, that
\[b^{(t)}=(a^{(t)}_1,...,\min\{a^{(t)}_i,k_i\},...,a^{(t)}_n).\]
Then for $j\neq i$, $b^{(t+1)}_j=f^U_j(b^{(t)})=h^{U^\prime}_j(a^{(t)})$ with probability at least $\frac{1}{2}$, and $b^{(t+1)}_i=f^U_i(b^{(t)})=\min\{h^{U^\prime}_j(a^{(t)}),k_i\}$ with probability at least $\frac{1}{2}$. So with positive probability,
\[b^{(t+1)}=(a^{(t+1)}_1,...,\min\{a^{(t+1)}_i,k_i\},...,a^{(t+1)}_n).\]
By induction, $\Pr[b^{(T)}=q^\prime]>0$, i.e., $q^\prime$ is $f^U$-reachable from $a$. We conclude that $f$ succeeds on $\U(A)$, and by Observation \ref{obs:5} it follows that $h$ succeeds on $\U(A)$.
\end{proof}

\newtheorem*{lem-9}{Theorem~\ref{lem:9}}
\begin{lem-9}
Let $n\geq 2$ and $k_1,...,k_n,k_{n+1}\geq2$. If $h$ succeeds on $\U(k_1,...,k_n,k_{n+1})$, then $h$ succeeds on $\U(k_1,...,k_i,...,k_n)$. The same is true if we replace $\U$ with $\G$.
\end{lem-9}
\begin{proof}
Let
\begin{align*}
A&=\{1,...,k_1\}\times...\times\{1,...,k_n\},\\
A^\prime&=\{1,...,k_1\}\times...\times\{1,...,k_n\}\times\{1,...,k_{n+1}\}.
\end{align*}
Suppose that $h$ succeeds on $\U(A^\prime)$, and for each $U=(u_1,...,u_n)\in\U(A)$, define a game $U^\prime=(u^\prime_1,...,u^\prime_n,u^\prime_{n+1})\in\U(A^\prime)$ such that for every $x=(x_1,...,x_n,x_{n+1})\in A^\prime$,
\[u^\prime_i(x)=u_i((x_1,...,x_n))\]
for each player $i\in\{1,...,n\}$ and
\[
u^\prime_{n+1}(x)=\left\{\begin{array}{ll}
1\;&\;\mbox{if }x_{n+1}=1\\
0\;&\;\mbox{otherwise}.
\end{array}\right.
\]
Informally, the first $n$ players are apathetic about player $n+1$'s action, and player $n+1$ always prefers to play $1$. Notice that $x=(x_1,...,x_n)\in A$ is a PNE for $U$ if and only if $(x_1,...,x_n,1)$ is a PNE for $U^\prime$.

Given a game $U$, we use $U^\prime$ to define a strategy mapping $f$ for games on $A$. For each $x=(x_1,...,x_n)\in A$ and $i\in\{1,...,n\}$,
\[f^U_i(x)=h^{U^\prime}_i((x_1,...,x_n,1)).\]

Now fix $U=(u_1,...,u_n)\in\U(A)$ and $a=(a_1,...,a_n)\in A$, and assume that $U$ has at least one pure Nash equilibrium. Then $U^\prime$ does also, so letting $a^\prime=(a_1,...,a_n,1)\in A^\prime$, some PNE $p^\prime=(p_1,...,p_n,1)$ for $U^\prime$ is $h^{U^\prime}$-reachable from $a^\prime$. We show that $p=(p_1,...,p_n)$, which is a PNE for $U$, is $f^U$-reachable from $a$.

Since $p^\prime$ is $h^{U^\prime}$-reachable from $a^\prime$, there is a partial $h^{U^\prime}$-run $a^{(0)},...,a^{(T)}$, for some $T\in\N$, such that $a^{(0)}=a^\prime$ and $a^{(T)}=p^\prime$. For each $t\in\{0,...,T-1\}$, if $a^{(t)}_{n+1}=1$, then player $n+1$ is best-replying at $a^{(t)}$, so $a^{(t+1)}_{n+1}=1$. Thus player 1 is playing 1 at every state in the partial run. Now let $b^{(0)},...,b^{(T)}$ be a partial $f^U$-run such that $b^{(0)}=a$.  At each step $t$, if $a^{(t)}=(b^{(t)}_1,....,b^{(t)}_n,1)$, then
\[f^U(b^{(t)})=h^{U^\prime}(b^{(t)}_1,....,b^{(t)}_n,1)=h^{U^\prime}(b^{(t)}),\]
so w.p.p. $a^{(t+1)}=(b^{(t+1)}_1,....,b^{(t+1)}_n,1)$. It follows that w.p.p. $a^{(T)}=p^\prime$, i.e., $b^{(T)}=p$. Thus $p$ is $f^U$-reachable from $a$, so $f$ succeeds on $\U(A)$. By Observation \ref{obs:5}, then, $h$ succeeds on $\U(A)$.

For the second part of the lemma, simply notice that $U^\prime$ is generic whenever $U$ is, thus the above argument still holds when $\G(A)$ is substituted for $\U(A)$.
\end{proof}

\newtheorem*{lem-10}{Theorem~\ref{lem:10}}
\begin{lem-10}
No historyless uncoupled stationary strategy mapping succeeds on $\U(2,2,2)$.
\end{lem-10}
\begin{proof}
Let $A=\{1,2,3\}$. By Observation \ref{obs:5} it suffices to show that $h$ does not succeed on $\U(A)$. Consider the game $U=(u_1,u_2,u_3)\in\U(A)$ where $u_i((x,y,z))$ is the $i$th coordinate of $M_{x}[y,z]$, for
\[
M_1=\left[\begin{matrix}
1,1,1\;&\;1,0,1\\
1,0,0\;&\;0,1,1
\end{matrix}
\right]\;\;\;\;\;\;
M_2=\left[
\begin{matrix}
0,1,0\;&\;0,1,1\\
0,0,0\;&\;1,0,1
\end{matrix}
\right].
\]
The unique PNE of $U$ is $p=(1,1,1)$. Let $a\in A$ with $a_3=2$. Then $h^U_3(a)=2$, since $2\in BR_3(a)$ for every $a$. It follows that under $h^U$, if the third player initially plays $2$, then it will never play $1$, so $p$ is not $h^U$-reachable from, for example, $(1,1,2)$. Thus $h$ does not succeed on $\U(A)$. 
\end{proof}

\newtheorem*{claimfrom}{Claim from Theorem~\ref{thm:11}}
\begin{claimfrom}
If either player has a strictly dominant action in $U^0$ or $U^1$, then some PNE is $h^U$-reachable from $a$.
\end{claimfrom}
\begin{proof}
Suppose that player 1 has a strictly dominant action $\alpha$ in $U^0$, and consider five cases.
\begin{enumerate}
\item $U$ has a PNE $(p_1,p_2,0)$, and $a_3=0$. Then player 1 is best-replying at $a$ only if $a_1=\alpha=p_1$, so w.p.p. $h^U(a)=(p_1,a_2,0)$. Player 2 is best-replying at $(p_1,a_2,0)$ only if $a_2=p_2$, so w.p.p. $h^U(h^U(a))=(p_1,p_2,0)$.
\item $U$ has a PNE $(q_1,q_2,1)$, $a_3=1$, and $BR^U_3(a)=1$.
\subitem If some player has a strictly dominant action in $U^1$, then this is symmetric to the situation described in case 1, and $q$ is $h^U$-reachable from $a$.
\subitem So assume that no player has a strictly dominant action in $U^1$. If $a$ is not a PNE for $U$, then $(a_1,a_2)$ is not a PNE for $U^1$. So by Lemma $\ref{lem:12}$, $(q_1,q_2)$ is $h^{U^1}$-reachable from $(a_1,a_2)$, i.e., the PNE $(q_1,q_2,1)$ is $h^U$-reachable from $a$.
\item $U$ has no PNE $(p_1,p_2,0)$, and $a_3=0$. As in case 1, w.p.p. $h^U(a)=(\alpha,a_2,0)$. Let $b_2=BR_2(\alpha,a_2,0)$. Then w.p.p. $h^U((\alpha,a_2,0))=(\alpha,b_2,0)$, and player 3 is not best-replying at $(\alpha,b_2,0)$ since it is not a PNE for $U$. Thus letting $b=(\alpha,b_2,1)$, w.p.p. $h^U((\alpha,b_2,0))=b$, so $b$ is $h^U$-reachable from $a$, and $b$ is an instance of case 2.
\item $U$ has a PNE $(q_1,q_2,1)$, $a_3=1$, and $BR^U_3(a)=0$.  Then w.p.p. $h^U(a)=(a_1,a_2,0)$, which is an instance of case 1 or 3. 
\item $U$ has no PNE $(q_1,q_2,1)$, and $a_3=1$.
\subitem If some player has a strictly dominant action in $U^1$, then w.p.p. that action will be played in $h^U(a)$, and w.p.p. the other player will play its best reply to that action in the next stage. Then the first two players are playing a PNE for $U^1$, so player 3 is not best-replying and may play 0 in the next round, giving an instance of case 1.
\subitem So assume that no player has a strictly dominant action in $U^1$. There is some $(b_1,b_2)\in A^\prime$ such that $BR_3(b_1,b_2,1)=0$, so w.p.p. $h^U(b_1,b_2,1)=(b_1,b_2,0)$. If $(a_1,a_2)$ is a PNE for $U^1$, then player 3 is not best replying and w.p.p. $h^U(a)=(a_1,a_2,0)$. Otherwise by Lemma \ref{lem:12} $(b_1,b_2)$ is $h^{U^1}$-reachable from $(a_1,a_2)$, so $(b_1,b_2,0)$, which is an instance of case 1, is $h^U$-reachable from $a$.
\end{enumerate}
It follows that some PNE for $U$ is $h^U$-reachable from every $a\in A$. By symmetry, the same holds whenever either player has a strictly dominant action in either $U^0$ or $U^1$.

\end{proof}

\newtheorem*{lem-12}{Theorem~\ref{lem:12}}
\begin{lem-12}
Let $k,l\in\N$, and let $U\in\G(k,\ell)$ be a game in which neither player has a strictly dominant action. For every $a,b\in A$ such that $a$ is not a PNE for $U$, $b$ is $h^U$-reachable from $a$.
\end{lem-12}
\begin{proof}
If $k=\ell=2$, then each player either prefers to match or to mismatch the other's action, and lemma holds by routine inspection of the four possibilities. So assume $\ell>2$.

Let $a,b\in A$, where $a$ is not a PNE for $U$. Notice that because $U$ is generic, $A$ contains exactly $\ell$ states where player 1 is best-replying and $k$ states where player 2 is best-replying, so there are at most $k+\ell$ states where either player is best-replying. And for any $x,y\in A$, if neither player is best-replying at $x$, then $h^U(x)=y$ with probability $\frac{1}{k\ell}$. Hence it suffices to show that more than $k+\ell$ distinct states in $A$ are reachable from $a$.

If player 2 is best-replying at $a$, then since player 2 has no dominant action, player 1 has some action $a^\prime_1$ such that player 2 is not best-replying at $(a^\prime_1,a_2)$. And player 1 is not best replying at $a$ (since $a$ is not a PNE), so w.p.p. $h^U(a)=(a^\prime_1,a_2)$. Thus some state in which player 2 is not best-replying is reachable from $a$.

Let $x=(\alpha,\beta)$ be such a state and consider the number of distinct states reachable from $x$. Player 2 might play any of its actions, so there are at least the $\ell$ possibilities $(\alpha,1),...,(\alpha,\ell)$ for $h^U(x)$. Since player 1 has no dominant action, there is some $\gamma\in\{1,...,\ell\}$ such that, letting $y=(\alpha,\gamma)$, $\alpha\not\in BR_1(y)$, so player 1 is not best-replying at $y$. By the same logic, $(1,\gamma),...,(k,\gamma)$ are possibilities for $h^U(y)$, and there is a $z=(\delta,\gamma)$ such that player 2 is not best-replying at $z$ and $(\delta,1),...,(\delta,\ell)$ are possibilities for $h^U(y)$.

We've shown that $(\alpha,1),...,(\alpha,\ell),(1,\gamma),...,(k,\gamma),(\delta,1),...,(\delta,\ell)$ are all reachable from $x$. Suppose that $\alpha=\delta$. Then $y=z$ and neither player is best replying at $y$, so all of $A$ is reachable from $y$. Otherwise,
\begin{align*}
\big|\{(\alpha,1),...,(\alpha,\ell),(1,\gamma),...,(k,\gamma),(\delta,1),...,(\delta,\ell)\}\big|&\geq k+2\ell-2\\
&>k+\ell.
\end{align*}
Since $x$ and $y$ are both reachable from $a$, this completes the proof.
\end{proof}

\newtheorem*{lem-14}{Theorem~\ref{lem:14}}
\begin{lem-14}
For every $k,\ell\geq 2$, $h$ does not succeed on $\G(2,2,k,\ell)$.
\end{lem-14}
\begin{proof}
Let $A=\{1,2\}\times \{1,2\}\times \{1,...,k_3\}\times\{1,...,k_4\}$, with $k_3,k_4\geq 2$. By Observation \ref{obs:5}, it suffices to show that $h$ does not succeed on $\G(A)$. Let $U=(u_1,u_2,u_3,u_4)\in\G(A)$ be defined as follows. For every $a=(a_1,a_2,a_3,a_4)\in A$,
\begin{align*}
u_1(a)&=\left\{\begin{array}{ll}
1\;&\;\mbox{if }a_1=a_2\\
0\;&\;\mbox{otherwise,}
\end{array}\right.\\
u_2(a)&=\left\{\begin{array}{ll}
1\;&\;\mbox{if }a_1=a_2\mbox{ XOR }a_3=a_4=1\\
0\;&\;\mbox{otherwise,}
\end{array}\right.\\
u_3(a)=u_4(a)&=\left\{\begin{array}{ll}
1\;&\;\mbox{if }a_3=a_4\\
0\;&\;\mbox{otherwise.}
\end{array}\right.
\end{align*}
Informally, player 1 always wants to match player 2's action, players 3 and 4 always want to match each other's actions, and player 2 wants to match player 1's action except when players 3 and 4 are both playing $1$, in which case player 2 wants to \emph{mismatch} player 1's action.

$U$ has the unique PNE $(2,2,2,2)$. Let $a\in A$ such that $a_3=a_4$. Then players 3 and 4 are both best-replying, so $h^U_3(a)=h^U_4(a)=1$. It follows that $(2,2,2,2)$ is not $h^U$-reachable from $(1,1,1,1)$, so $h$ does not succeed on $\G(A)$.
\end{proof}

\newtheorem*{lem-15}{Theorem~\ref{lem:15}}
\begin{lem-15}
For every $k_1,k_2,k_3\geq 3$, $h$ does not succeed on $\G(k_1,k_2,k_3)$
\end{lem-15}
\begin{proof}
Let $A=\{1,...,k_1\}\times \{1,...,k_2\}\times \{1,...,k_3\}$, with $k_1,k_2,k_3\geq 3$.
By Observation \ref{obs:5}, it suffices to show that $h$ does not succeed on $\G(A)$. Hart and Mas-Colell~\cite{HarMas06} give an example of a 3-by-3-by-3 generic game on which no historyless uncoupled strategy mapping succeeds. The game is $U=(u_1,u_2,u_3)\in\U(\{1,2,3\}^3)$ where $u_i((x,y,z))$ is the $i$th coordinate of $M_x[y,z]$, for
\begin{align*}
M_1&=\begin{bmatrix}
0,0,0\;&\;0,4,4\;&\;2,1,2\\
4,4,0\;&\;4,0,4\;&\;3,1,3\\
1,2,3\;&\;1,3,3\;&\;0,0,0
\end{bmatrix}\\
M_2&=\begin{bmatrix}
4,0,4\;&\;4,4,0\;&\;3,1,3\\
0,4,4\;&\;0,0,0\;&\;2,1,2\\
1,3,3\;&\;1,2,2\;&\;0,0,0
\end{bmatrix}\\
M_3&=\begin{bmatrix}
2,2,1\;&\;3,3,1\;&\;0,0,0\\
3,3,1\;&\;2,2,1\;&\;0,0,0\\
0,0,0\;&\;0,0,0\;&\;6,6,6
\end{bmatrix}.
\end{align*}
They observe that $U$ has the unique PNE $(3,3,3)$, and prove that if $a\in A$ contains both a $1$ and a $2$, then for any uncoupled historyless strategy mapping $f$, $f^U(a)$ also contains both a $1$ and a $2$. To prove the lemma, we pad the game with extra actions and show that the expanded game retains this property.

We define the expanded game $U^\prime=(u^\prime_1,u^\prime_2,u^\prime_3)\in\G(A)$ by, for each player $i$ and profile $a=(a_1,a_2,a_3)\in A$, \[
u^\prime_i((a_1,a_2,a_3))=\left\{\begin{array}{ll}0\;&\;\mbox{if }a_i>3\\
u_i((\min\{a_1,3\},\min\{a_2,3\},\min\{a_3,3\}))\;&\;\mbox{otherwise.}
\end{array}\right.
\]
So for each new action $a_i>3$ for player $i$, $a_i$ is weakly dominated and both other players are indifferent to whether $i$ plays $a_i$ or $3$.

Suppose that at $a\in A$ at least one player is playing $1$ and at least one player is playing $2$. If $a\in\{1,2,3\}^3$, then since all the new actions are weakly dominated, Hart and Mas-Colell's analysis applies directly: a player playing $1$ and a player playing $2$ are best-replying, so $h^U(a)$ contains both a $1$ and a $2$. Otherwise, one player $i$ is playing $a_i>3$. In this case the other two players are best-replying, and they played $1$ and $2$, so $h^U(a)$ again contains both a $1$ and a $2$. It follows that the players will never reach the PNE $(3,3,3)$ starting from, for example, $(1,2,1)$, when following $h^U$.
\end{proof}

\newtheorem*{thm-18}{Theorem~\ref{thm:18}}
\begin{thm-18}
For every profile space $A$, no deterministic uncoupled historyless stationary strategy mapping succeeds on all games $(A,U)$.
\end{thm-18}
\begin{proof}
Except when $A=A_1\times A_2$ and either $|A_1|$ or $|A_2|$ is $2$, this follows directly from Theorem \ref{thm:7}. So let $k\geq 2$ and $A=\{1,2\}\times\{1,...,k\}$, and assume that some deterministic historyless uncoupled strategy mapping $f$ succeeds on $\U(A)$.

Consider the game $U=(u_1,u_2)\in\U(A)$ defined by
\begin{align*}
u_1(a)&=\left\{\begin{array}{ll}
1\;&\;\mbox{if }a_1=1\\
0\;&\;\mbox{if }a_1=2
\end{array}\right.\\
u_2(a)&=\left\{\begin{array}{ll}
1\;&\;\mbox{if }a_2=a_1=1\mbox{ or }a_2\geq a_1=2\\
0\;&\;\mbox{otherwise},
\end{array}\right.
\end{align*}
for every $a=(a_1,a_2)\in A$. The unique PNE of this game is $p=(1,1)$, so since we assumed that $f$ succeeds on $\U(A)$, $p$ is $f^U$-reachable from every $a\in A$.

Define a new game $U^\prime=(u^\prime_1,u^\prime_2)\in\U(A)$ by
\begin{align*}
u^\prime_1(b)&=\left\{\begin{array}{ll}
2\;&\;\mbox{if }b_2\geq x_1=2\mbox{ and }f^U_2(1,b_2)=1\\
u_1(b)\;&\;\mbox{otherwise}
\end{array}\right.\\
u^\prime_2(b)&=u_2(b),
\end{align*}
for every $b=(b_1,b_2)\in A$. Informally, each player's preferences are exactly the same as in $U$, except that player 1 now prefers to play $2$ whenever $f^U$ would instruct player 2 to play $1$. Notice that $U^\prime$ also has $p=(1,1)$ as its unique PNE, and that by uncoupledness, $f^{U^\prime}_2(b)=f^U_2(b)$ for every $b\in A$.

Let $a=(1,\alpha)\in A$, for some $\alpha\neq 1$. Notice that $u^\prime_1(a)=u_1(a)=1$, and consider two cases.

\begin{enumerate}
\item $f^U_2(a)=1$. Then $u^\prime_1((2,\alpha))=2$, so player 1 is not $U^\prime$-best-replying at $a$. Thus by Observation \ref{obs:4} $f^{U^\prime}_1(a)\neq 1$. Since $f^{U^\prime}_2(a)=f^U_2(a)=1$, we have $f^{U^\prime}(a)=(2,1)$.
\item $f^U_2(a)\neq 1$. Then $u^\prime_1((2,\alpha))=u_1((2,\alpha))=0$, so player 1 is $U^\prime$-best-replying at $a$, so by Observation \ref{obs:4}, $f^{U^\prime}_1(a)=1$. Since $f^{U^\prime}_2(a)=f^U_2(a)\neq1$, we have $f^{U^\prime}(a)=(1,\beta)$ for some $\beta\neq1$.
\end{enumerate}
Now let $b=(2,1)$. Then $u^\prime_1(b)=u_1(b)=0$, and $u^\prime_2(b)=u_2(b)=0$, so neither player is best-replying. So by Observation \ref{obs:4} $f^{U^\prime}_1(b)\neq 2$ and $f^{U^\prime}_2(b)\neq 1$, i.e., $f^{U^\prime}(a)=(1,\beta)$ for some $\beta\neq1$. It follows that $p=(1,1)$ is not $f^{U^\prime}$-reachable from $(2,1)$, so $f$ does not guarantee convergence to a PNE in $U^\prime$, hence $f$ does not succeed on $\U(A)$.
\end{proof}

\end{document}